\documentclass[11pt]{article}

\usepackage[T1]{fontenc}
\usepackage[utf8]{inputenc}
\usepackage{amsmath,amssymb,amsthm,mathtools}
\usepackage{geometry}
\usepackage{hyperref}

\newtheorem{theorem}{Theorem}
\newtheorem{proposition}[theorem]{Proposition}
\newtheorem{corollary}[theorem]{Corollary}

\theoremstyle{remark}
\newtheorem{remark}[theorem]{Remark}

\newcommand{\R}{\mathbb{R}}
\newcommand{\Tr}{\operatorname{Tr}}
\newcommand{\im}{\operatorname{im}}
\newcommand{\Ker}{\operatorname{ker}}
\newcommand{\doi}[1]{\href{https://doi.org/#1}{doi:#1}}

\title{
Information Loss under Gauge Reduction\\
of Discrete Electromagnetic Fields
}

\author{
Jean-Pierre Magnot\\
SFR MATHSTIC, LAREMA, Université d'Angers, France\\
Lepage Research Institute, Prešov, Slovakia\\
\texttt{magnot@math.cnrs.fr}
}

\date{}

\begin{document}

\maketitle

\begin{abstract}
We study relative entropy under gauge reduction in
Whitney-discretized electromagnetism. The discrete Hodge
decomposition separates exact gauge directions from harmonic and
coexact physical degrees of freedom. For Gaussian distributions on
the auxiliary potential space, relative entropy splits exactly into
the relative entropy of the physical marginals and a nonnegative
conditional gauge contribution. We express this contribution through
Schur complements, conditional means, and physical--gauge
correlations, and characterize equality between auxiliary and
physical relative entropies. The physical divergence is also shown
to be the minimum over all Gaussian auxiliary extensions. At the
infinitesimal level, Fisher information admits an analogous
decomposition. A two-variable model illustrates how statistically
distinct auxiliary potentials may represent identical physical
statistics. The construction provides an information-geometric
interpretation of discrete gauge reduction without assigning
physical significance to the eliminated gauge variables.
\end{abstract}

\medskip

\noindent
\textbf{Keywords:}
discrete electromagnetism; gauge reduction; Whitney forms;
Gaussian states; relative entropy; Fisher information.

\medskip

\noindent
\textbf{2020 Mathematics Subject Classification:}
81T13; 94A17; 65N30; 58A12.

\section{Introduction}

Gauge potentials contain more variables than are required to
describe electromagnetic observables. In the Abelian setting, two
potentials related by
\begin{equation}
A\longmapsto A+d\chi
\label{eq:continuous-gauge}
\end{equation}
determine the same electromagnetic field.

This redundancy plays a central role in both classical field theory
and covariant quantization. From a statistical perspective, it
raises a natural question: how much distinguishability between two
auxiliary descriptions disappears when unobservable gauge variables
are removed?

The answer is especially transparent for compatible discrete
electromagnetic models. Whitney forms and finite element exterior
calculus preserve the de Rham complex at the discrete level
\cite{ArnoldFalkWinther2006,ArnoldFalkWinther2010}. Their
application to electromagnetic systems provides gauge-compatible
discretizations based on simplicial incidence operators and
discrete Hodge decompositions
\cite{Bossavit1998,Stern2015,Desbrun2005}.

In this framework, the space of discrete potentials decomposes as
\begin{equation}
V
=
V_{\mathrm p}\oplus V_{\mathrm g},
\label{eq:intro-splitting}
\end{equation}
where \(V_{\mathrm g}\) consists of exact gauge components and
\(V_{\mathrm p}\) contains the harmonic and coexact physical
components.

For two probability distributions \(\mu,\nu\) on \(V\), let
\begin{equation}
\mu_{\mathrm p},
\qquad
\nu_{\mathrm p}
\end{equation}
denote their physical marginals.

The central identity is
\begin{equation}
D(\mu\|\nu)
=
D(\mu_{\mathrm p}\|\nu_{\mathrm p})
+
\int_{V_{\mathrm p}}
D\left(
\mu_{\mathrm g|x}
\middle\|
\nu_{\mathrm g|x}
\right)
\,d\mu_{\mathrm p}(x),
\label{eq:intro-chain-rule}
\end{equation}
where the second term compares the conditional gauge
distributions.

The identity is a particular instance of the classical chain rule
for relative entropy \cite{CoverThomas2006,Kullback1951}. Its
contribution here is its explicit realization on the Whitney
electromagnetic complex, including closed formulas for Gaussian
fields and a characterization of physically equivalent auxiliary
descriptions.

In particular, two auxiliary distributions may satisfy
\begin{equation}
D(\mu\|\nu)>0
\end{equation}
while
\begin{equation}
D(\mu_{\mathrm p}\|\nu_{\mathrm p})=0.
\end{equation}
Their distinguishability then originates entirely from the
unphysical gauge representation.

Throughout, we consider classical Gaussian probability measures on
finite-dimensional discrete potential spaces. No indefinite metric,
quantum density operator, or continuum limit is required.

\section{Whitney discretization and physical reduction}

Let \(M\) be a closed oriented Riemannian three-manifold, and let
\(K\) be a finite triangulation.

Write
\begin{equation}
C^k(K)
=
C^k(K;\R)
\end{equation}
for simplicial \(k\)-cochains and
\begin{equation}
\delta_k:
C^k(K)\longrightarrow C^{k+1}(K)
\end{equation}
for the coboundary.

The cochain complex satisfies
\begin{equation}
\delta_{k+1}\delta_k=0.
\label{eq:delta-square}
\end{equation}

The Whitney maps
\begin{equation}
W_k:
C^k(K)\longrightarrow\Omega^k_{\mathrm{pw}}(M)
\end{equation}
satisfy
\begin{equation}
dW_k=W_{k+1}\delta_k.
\label{eq:whitney}
\end{equation}

A discrete electromagnetic potential is
\begin{equation}
a\in C^1(K),
\end{equation}
and its discrete magnetic field is
\begin{equation}
b=\delta_1a.
\label{eq:field}
\end{equation}

Gauge transformations are
\begin{equation}
a\longmapsto a+\delta_0\chi,
\qquad
\chi\in C^0(K).
\label{eq:gauge}
\end{equation}

Since
\begin{equation}
\delta_1\delta_0=0,
\end{equation}
the field \eqref{eq:field} is unchanged.

We equip each cochain space with the Whitney inner product
\begin{equation}
\langle u,v\rangle_k
=
\int_M
W_k(u)\wedge\star W_k(v).
\label{eq:whitney-inner}
\end{equation}

Let
\begin{equation}
\delta_k^\ast
\end{equation}
be the associated adjoint.

The degree-one harmonic space is
\begin{equation}
\mathcal H_K^1
=
\Ker\delta_1
\cap
\Ker\delta_0^\ast.
\label{eq:harmonic}
\end{equation}

The discrete Hodge decomposition reads
\begin{equation}
C^1(K)
=
\im\delta_0
\oplus
\mathcal H_K^1
\oplus
\im\delta_1^\ast.
\label{eq:hodge}
\end{equation}

Define
\begin{equation}
V_{\mathrm g}
=
\im\delta_0
\label{eq:gauge-sector}
\end{equation}
and
\begin{equation}
V_{\mathrm p}
=
\mathcal H_K^1
\oplus
\im\delta_1^\ast.
\label{eq:physical-sector}
\end{equation}

Then
\begin{equation}
C^1(K)
=
V_{\mathrm p}\oplus V_{\mathrm g}.
\label{eq:orthogonal-splitting}
\end{equation}

The physical projection is
\begin{equation}
P_{\mathrm p}:
C^1(K)\longrightarrow V_{\mathrm p}.
\label{eq:projection}
\end{equation}

Because \(P_{\mathrm p}\) annihilates
\(\im\delta_0\),
\begin{equation}
P_{\mathrm p}(a+\delta_0\chi)
=
P_{\mathrm p}(a).
\label{eq:projection-gauge}
\end{equation}

Moreover,
\begin{equation}
C^1(K)/\im\delta_0
\simeq
V_{\mathrm p}.
\label{eq:quotient}
\end{equation}

\begin{remark}
Harmonic components are retained in \(V_{\mathrm p}\). Although
their magnetic curvature vanishes, they are not exact gauge
directions and may encode physical holonomy information in
situations where the corresponding observables are available.
The physical projection used here therefore represents reduction
modulo exact gauge transformations, rather than restriction to
curvature observables alone.
\end{remark}

Choose Whitney-orthonormal coordinates
\begin{equation}
a=(x,y),
\qquad
x\in V_{\mathrm p},
\quad
y\in V_{\mathrm g}.
\label{eq:coordinates}
\end{equation}

Set
\begin{equation}
p=\dim V_{\mathrm p},
\qquad
q=\dim V_{\mathrm g}.
\label{eq:dimensions}
\end{equation}

\section{Relative entropy and conditional gauge information}

Let \(\mu,\nu\) be nondegenerate Gaussian probability measures
on
\begin{equation}
V_{\mathrm p}\oplus V_{\mathrm g}.
\end{equation}

Their physical marginals are
\begin{equation}
\mu_{\mathrm p}
=
(P_{\mathrm p})_\#\mu,
\qquad
\nu_{\mathrm p}
=
(P_{\mathrm p})_\#\nu.
\label{eq:marginals}
\end{equation}

Here
\begin{equation}
(P_{\mathrm p})_\#
\end{equation}
denotes the pushforward under the physical projection.

Let
\begin{equation}
\mu_{\mathrm g|x},
\qquad
\nu_{\mathrm g|x}
\end{equation}
be the corresponding conditional gauge distributions.

For probability densities \(f,g\), relative entropy is
\begin{equation}
D(f\|g)
=
\int
f\log\frac{f}{g}.
\label{eq:relative}
\end{equation}

\begin{theorem}
\label{thm:chain}
The relative entropy satisfies
\begin{equation}
D(\mu\|\nu)
=
D(\mu_{\mathrm p}\|\nu_{\mathrm p})
+
\mathcal I_{\mathrm g}(\mu,\nu),
\label{eq:chain}
\end{equation}
where
\begin{equation}
\mathcal I_{\mathrm g}(\mu,\nu)
=
\int_{V_{\mathrm p}}
D\left(
\mu_{\mathrm g|x}
\middle\|
\nu_{\mathrm g|x}
\right)
\,d\mu_{\mathrm p}(x).
\label{eq:gauge-information}
\end{equation}

In particular,
\begin{equation}
\mathcal I_{\mathrm g}(\mu,\nu)\geq0
\label{eq:nonnegative}
\end{equation}
and
\begin{equation}
D(\mu_{\mathrm p}\|\nu_{\mathrm p})
\leq
D(\mu\|\nu).
\label{eq:contraction}
\end{equation}

Equality holds if and only if
\begin{equation}
\mu_{\mathrm g|x}
=
\nu_{\mathrm g|x}
\label{eq:conditional-equality}
\end{equation}
for \(\mu_{\mathrm p}\)-almost every \(x\).
\end{theorem}

\begin{proof}
Let
\begin{equation}
f(x,y)
=
f_{\mathrm p}(x)f_{\mathrm g|x}(y)
\end{equation}
and
\begin{equation}
g(x,y)
=
g_{\mathrm p}(x)g_{\mathrm g|x}(y)
\end{equation}
be the respective Gaussian densities.

Then
\begin{equation}
\log\frac{f(x,y)}{g(x,y)}
=
\log\frac{
f_{\mathrm p}(x)
}{
g_{\mathrm p}(x)
}
+
\log\frac{
f_{\mathrm g|x}(y)
}{
g_{\mathrm g|x}(y)
}.
\end{equation}

Integrating against \(f(x,y)\,dx\,dy\), the first term becomes
\begin{equation}
D(\mu_{\mathrm p}\|\nu_{\mathrm p}),
\end{equation}
while the second becomes
\begin{equation}
\int_{V_{\mathrm p}}
D\left(
\mu_{\mathrm g|x}
\middle\|
\nu_{\mathrm g|x}
\right)
\,d\mu_{\mathrm p}(x).
\end{equation}

This proves \eqref{eq:chain}. The remaining statements follow
from nonnegativity of relative entropy and its vanishing
criterion.
\end{proof}

\begin{remark}
The quantity \(\mathcal I_{\mathrm g}(\mu,\nu)\) is not an
intrinsic physical observable. It depends on the auxiliary
extensions \(\mu,\nu\). Its meaning is precisely the amount
of distinguishability removed by passing from auxiliary
descriptions to the physical gauge quotient.
\end{remark}

\section{Explicit Gaussian formula}

Write
\begin{equation}
\mu
=
\mathcal N(m_\mu,\Sigma_\mu),
\qquad
\nu
=
\mathcal N(m_\nu,\Sigma_\nu),
\label{eq:gaussians}
\end{equation}
with
\begin{equation}
m_\mu
=
\begin{pmatrix}
m_{\mu,\mathrm p}\\
m_{\mu,\mathrm g}
\end{pmatrix},
\qquad
m_\nu
=
\begin{pmatrix}
m_{\nu,\mathrm p}\\
m_{\nu,\mathrm g}
\end{pmatrix},
\label{eq:means}
\end{equation}
and
\begin{equation}
\Sigma_\mu
=
\begin{pmatrix}
A_\mu&B_\mu\\
B_\mu^\top&C_\mu
\end{pmatrix},
\qquad
\Sigma_\nu
=
\begin{pmatrix}
A_\nu&B_\nu\\
B_\nu^\top&C_\nu
\end{pmatrix}.
\label{eq:blocks}
\end{equation}

The physical marginals are
\begin{equation}
\mu_{\mathrm p}
=
\mathcal N(m_{\mu,\mathrm p},A_\mu),
\label{eq:mu-physical}
\end{equation}
and
\begin{equation}
\nu_{\mathrm p}
=
\mathcal N(m_{\nu,\mathrm p},A_\nu).
\label{eq:nu-physical}
\end{equation}

Define
\begin{equation}
R_\mu
=
B_\mu^\top A_\mu^{-1},
\qquad
R_\nu
=
B_\nu^\top A_\nu^{-1},
\label{eq:R}
\end{equation}
and the Schur complements
\begin{equation}
S_\mu
=
C_\mu-B_\mu^\top A_\mu^{-1}B_\mu,
\label{eq:Smu}
\end{equation}
\begin{equation}
S_\nu
=
C_\nu-B_\nu^\top A_\nu^{-1}B_\nu.
\label{eq:Snu}
\end{equation}

Since \(\Sigma_\mu,\Sigma_\nu\) are positive definite,
\begin{equation}
S_\mu>0,
\qquad
S_\nu>0.
\end{equation}

The conditional gauge laws are
\begin{equation}
\mu_{\mathrm g|x}
=
\mathcal N(r_\mu(x),S_\mu)
\label{eq:conditional-mu}
\end{equation}
and
\begin{equation}
\nu_{\mathrm g|x}
=
\mathcal N(r_\nu(x),S_\nu),
\label{eq:conditional-nu}
\end{equation}
where
\begin{equation}
r_\mu(x)
=
m_{\mu,\mathrm g}
+
R_\mu(x-m_{\mu,\mathrm p})
\label{eq:rmu}
\end{equation}
and
\begin{equation}
r_\nu(x)
=
m_{\nu,\mathrm g}
+
R_\nu(x-m_{\nu,\mathrm p}).
\label{eq:rnu}
\end{equation}

\begin{theorem}
\label{thm:explicit}
Set
\begin{equation}
\Delta R
=
R_\mu-R_\nu
\label{eq:delta-R}
\end{equation}
and
\begin{equation}
d
=
m_{\mu,\mathrm g}
-
m_{\nu,\mathrm g}
-
R_\nu
\left(
m_{\mu,\mathrm p}
-
m_{\nu,\mathrm p}
\right).
\label{eq:d}
\end{equation}

Then
\begin{align}
\mathcal I_{\mathrm g}(\mu,\nu)
=
\frac12
\Bigg[
&
\Tr(S_\nu^{-1}S_\mu)
-
q
+
\log
\frac{
\det S_\nu
}{
\det S_\mu
}
\nonumber\\
&+
d^\top S_\nu^{-1}d
\nonumber\\
&+
\Tr\left(
S_\nu^{-1}
\Delta R
A_\mu
\Delta R^\top
\right)
\Bigg].
\label{eq:explicit-gauge-info}
\end{align}
\end{theorem}

\begin{proof}
For nondegenerate Gaussian measures on \(\R^q\),
\begin{align}
&D\left(
\mathcal N(r_1,S_1)
\middle\|
\mathcal N(r_2,S_2)
\right)
\nonumber\\
&\quad=
\frac12
\left[
\Tr(S_2^{-1}S_1)
-
q
+
\log
\frac{
\det S_2
}{
\det S_1
}
+
(r_1-r_2)^\top
S_2^{-1}
(r_1-r_2)
\right].
\label{eq:gaussian-relative}
\end{align}

By \eqref{eq:rmu} and \eqref{eq:rnu},
\begin{equation}
r_\mu(x)-r_\nu(x)
=
d+\Delta R(x-m_{\mu,\mathrm p}).
\label{eq:r-difference}
\end{equation}

Under \(\mu_{\mathrm p}\),
\begin{equation}
\mathbb E_{\mu_{\mathrm p}}
\left[
x-m_{\mu,\mathrm p}
\right]
=
0
\label{eq:centered}
\end{equation}
and
\begin{equation}
\mathbb E_{\mu_{\mathrm p}}
\left[
(x-m_{\mu,\mathrm p})
(x-m_{\mu,\mathrm p})^\top
\right]
=
A_\mu.
\label{eq:covariance}
\end{equation}

Therefore,
\begin{align}
&
\mathbb E_{\mu_{\mathrm p}}
\left[
(r_\mu(x)-r_\nu(x))^\top
S_\nu^{-1}
(r_\mu(x)-r_\nu(x))
\right]
\nonumber\\
&\qquad=
d^\top S_\nu^{-1}d
+
\Tr\left(
S_\nu^{-1}
\Delta R
A_\mu
\Delta R^\top
\right).
\label{eq:expectation-quadratic}
\end{align}

Integrating \eqref{eq:gaussian-relative} with respect to
\(\mu_{\mathrm p}\) gives
\eqref{eq:explicit-gauge-info}.
\end{proof}

Formula \eqref{eq:explicit-gauge-info} separates three sources
of auxiliary distinguishability:

\begin{enumerate}
\item
Differences between conditional gauge covariances, represented
by
\begin{equation}
\Tr(S_\nu^{-1}S_\mu)
-
q
+
\log
\frac{
\det S_\nu
}{
\det S_\mu
}.
\end{equation}

\item
Differences between conditional gauge means, represented by
\begin{equation}
d^\top S_\nu^{-1}d.
\end{equation}

\item
Differences between physical--gauge regression operators,
represented by
\begin{equation}
\Tr\left(
S_\nu^{-1}
\Delta R
A_\mu
\Delta R^\top
\right).
\end{equation}
\end{enumerate}

\begin{corollary}
\label{cor:equality}
One has
\begin{equation}
D(\mu\|\nu)
=
D(\mu_{\mathrm p}\|\nu_{\mathrm p})
\label{eq:no-loss}
\end{equation}
if and only if
\begin{equation}
S_\mu=S_\nu,
\label{eq:equal-S}
\end{equation}
\begin{equation}
R_\mu=R_\nu,
\label{eq:equal-R}
\end{equation}
and
\begin{equation}
m_{\mu,\mathrm g}
-
m_{\nu,\mathrm g}
=
R_\nu
\left(
m_{\mu,\mathrm p}
-
m_{\nu,\mathrm p}
\right).
\label{eq:equal-d}
\end{equation}
\end{corollary}

\begin{proof}
By Theorem \ref{thm:chain}, equality holds if and only if
\begin{equation}
\mu_{\mathrm g|x}
=
\nu_{\mathrm g|x}
\end{equation}
for \(\mu_{\mathrm p}\)-almost every \(x\).

Equality of the conditional Gaussian measures implies
\begin{equation}
S_\mu=S_\nu
\end{equation}
and
\begin{equation}
r_\mu(x)=r_\nu(x)
\end{equation}
for \(\mu_{\mathrm p}\)-almost every \(x\).

Because \(\mu_{\mathrm p}\) is nondegenerate and the
difference
\begin{equation}
r_\mu(x)-r_\nu(x)
\end{equation}
is affine in \(x\), this affine function must vanish
identically.

Its linear part gives
\begin{equation}
R_\mu=R_\nu,
\end{equation}
and its constant part gives
\eqref{eq:equal-d}.
\end{proof}

\section{Physical relative entropy as minimal auxiliary
distinguishability}

Let
\begin{equation}
\bar\mu,
\qquad
\bar\nu
\end{equation}
be nondegenerate Gaussian probability measures on
\(V_{\mathrm p}\).

Define
\begin{equation}
\operatorname{Ext}_{\mathrm G}(\bar\mu)
=
\left\{
\mu:
\mu \text{ is a nondegenerate Gaussian measure on }
V_{\mathrm p}\oplus V_{\mathrm g},
\;
(P_{\mathrm p})_\#\mu=\bar\mu
\right\}.
\label{eq:extensions}
\end{equation}

Define
\begin{equation}
\operatorname{Ext}_{\mathrm G}(\bar\nu)
\end{equation}
analogously.

\begin{proposition}
\label{prop:minimal}
The physical relative entropy satisfies
\begin{equation}
D(\bar\mu\|\bar\nu)
=
\min_{\substack{
\mu\in\operatorname{Ext}_{\mathrm G}(\bar\mu)
\\
\nu\in\operatorname{Ext}_{\mathrm G}(\bar\nu)
}}
D(\mu\|\nu).
\label{eq:minimal}
\end{equation}

A pair of Gaussian extensions attains the minimum if and only
if
\begin{equation}
\mu_{\mathrm g|x}
=
\nu_{\mathrm g|x}
\end{equation}
for \(\bar\mu\)-almost every \(x\).
\end{proposition}

\begin{proof}
For every admissible pair of extensions, Theorem
\ref{thm:chain} gives
\begin{equation}
D(\mu\|\nu)
=
D(\bar\mu\|\bar\nu)
+
\mathcal I_{\mathrm g}(\mu,\nu)
\geq
D(\bar\mu\|\bar\nu).
\end{equation}

Choose any nondegenerate Gaussian measure
\begin{equation}
\tau
\end{equation}
on \(V_{\mathrm g}\), and define
\begin{equation}
\mu=\bar\mu\otimes\tau,
\qquad
\nu=\bar\nu\otimes\tau.
\end{equation}
Then
\begin{equation}
\mu_{\mathrm g|x}
=
\nu_{\mathrm g|x}
=
\tau,
\end{equation}
and hence
\begin{equation}
D(\mu\|\nu)
=
D(\bar\mu\|\bar\nu).
\end{equation}

The equality characterization follows directly from Theorem
\ref{thm:chain}.
\end{proof}

Thus physical relative entropy is the least auxiliary
distinguishability compatible with the prescribed physical
states.

\section{Fisher information under gauge reduction}

Let
\begin{equation}
\theta
\longmapsto
\mu_\theta
\end{equation}
be a smooth family of nondegenerate Gaussian measures on
\begin{equation}
V_{\mathrm p}\oplus V_{\mathrm g}.
\end{equation}

Write its density as
\begin{equation}
f_\theta(x,y)
=
f_{\theta,\mathrm p}(x)
f_{\theta,\mathrm g|x}(y).
\label{eq:family-factorization}
\end{equation}

The auxiliary Fisher information matrix is
\begin{equation}
g^{\mathrm{aux}}_{ij}(\theta)
=
\int
\partial_i\log f_\theta
\,
\partial_j\log f_\theta
\,
f_\theta
\,dx\,dy.
\label{eq:aux-fisher}
\end{equation}

The physical Fisher matrix is
\begin{equation}
g^{\mathrm{phys}}_{ij}(\theta)
=
\int
\partial_i\log f_{\theta,\mathrm p}
\,
\partial_j\log f_{\theta,\mathrm p}
\,
f_{\theta,\mathrm p}
\,dx.
\label{eq:phys-fisher}
\end{equation}

\begin{proposition}
\label{prop:fisher}
The Fisher information decomposes as
\begin{equation}
g^{\mathrm{aux}}_{ij}(\theta)
=
g^{\mathrm{phys}}_{ij}(\theta)
+
g^{\mathrm{cond}}_{ij}(\theta),
\label{eq:fisher-decomposition}
\end{equation}
where
\begin{align}
g^{\mathrm{cond}}_{ij}(\theta)
&=
\int_{V_{\mathrm p}}
\left[
\int_{V_{\mathrm g}}
\partial_i
\log f_{\theta,\mathrm g|x}(y)
\right.
\nonumber\\
&\hspace{4em}
\left.
\times
\partial_j
\log f_{\theta,\mathrm g|x}(y)
\,
f_{\theta,\mathrm g|x}(y)
\,dy
\right]
f_{\theta,\mathrm p}(x)
\,dx.
\label{eq:conditional-fisher}
\end{align}

In particular,
\begin{equation}
g^{\mathrm{aux}}(\theta)
-
g^{\mathrm{phys}}(\theta)
\geq0
\label{eq:fisher-positive}
\end{equation}
as a quadratic form.
\end{proposition}

\begin{proof}
From \eqref{eq:family-factorization},
\begin{equation}
\partial_i\log f_\theta
=
\partial_i\log f_{\theta,\mathrm p}
+
\partial_i\log f_{\theta,\mathrm g|x}.
\end{equation}

Expanding the product in
\eqref{eq:aux-fisher} gives a physical term, a conditional
term, and two mixed terms.

For each fixed \(x\),
\begin{align}
\int_{V_{\mathrm g}}
\partial_i
\log f_{\theta,\mathrm g|x}(y)
\,
f_{\theta,\mathrm g|x}(y)
\,dy
&=
\int_{V_{\mathrm g}}
\partial_i
f_{\theta,\mathrm g|x}(y)
\,dy
\nonumber\\
&=
\partial_i
\int_{V_{\mathrm g}}
f_{\theta,\mathrm g|x}(y)
\,dy
\nonumber\\
&=
0.
\end{align}

Therefore both mixed terms vanish, yielding
\eqref{eq:fisher-decomposition}.

The conditional Fisher information is positive
semidefinite, proving \eqref{eq:fisher-positive}.
\end{proof}

The identity
\eqref{eq:fisher-decomposition} is the infinitesimal analogue
of the entropy decomposition \eqref{eq:chain}.

\section{A two-variable gauge model}

Consider
\begin{equation}
V_{\mathrm p}
=
\R,
\qquad
V_{\mathrm g}
=
\R.
\end{equation}

Let
\begin{equation}
x\sim\mathcal N(0,1)
\end{equation}
under both auxiliary descriptions.

Fix
\begin{equation}
s>0,
\end{equation}
and define
\begin{equation}
y|x
\sim
\mathcal N(rx,s^2)
\label{eq:first-conditional}
\end{equation}
under \(\mu\), while
\begin{equation}
y|x
\sim
\mathcal N(r'x,s^2)
\label{eq:second-conditional}
\end{equation}
under \(\nu\).

The joint covariance matrices are
\begin{equation}
\Sigma_\mu
=
\begin{pmatrix}
1&r\\
r&r^2+s^2
\end{pmatrix}
\label{eq:Sigma-mu}
\end{equation}
and
\begin{equation}
\Sigma_\nu
=
\begin{pmatrix}
1&r'\\
r'&(r')^2+s^2
\end{pmatrix}.
\label{eq:Sigma-nu}
\end{equation}

Both matrices have determinant
\begin{equation}
s^2>0.
\end{equation}

The physical marginals coincide:
\begin{equation}
\mu_{\mathrm p}
=
\nu_{\mathrm p}
=
\mathcal N(0,1).
\label{eq:same-physical}
\end{equation}
Therefore,
\begin{equation}
D(\mu_{\mathrm p}\|\nu_{\mathrm p})
=
0.
\label{eq:zero-physical}
\end{equation}

However,
\begin{equation}
D\left(
\mu_{\mathrm g|x}
\middle\|
\nu_{\mathrm g|x}
\right)
=
\frac{
(r-r')^2x^2
}{
2s^2
}.
\label{eq:conditional-example}
\end{equation}

Since
\begin{equation}
\mathbb E[x^2]=1,
\end{equation}
Theorem \ref{thm:chain} gives
\begin{equation}
\boxed{
D(\mu\|\nu)
=
\frac{
(r-r')^2
}{
2s^2
}.
}
\label{eq:example-result}
\end{equation}

Thus the entire auxiliary distinguishability is generated by
different correlations between the physical coordinate and
the gauge coordinate.

If \(r\neq r'\), then
\begin{equation}
D(\mu\|\nu)>0,
\end{equation}
although every observable depending only on \(x\) has
identical statistics under the two descriptions.

\section{Discussion and conclusion}

The discrete Hodge decomposition distinguishes exact gauge
directions from harmonic and coexact physical degrees of
freedom.

For Gaussian electromagnetic potentials, this distinction
induces an exact entropy balance:
\begin{equation}
\boxed{
D(\mu\|\nu)
=
D(\mu_{\mathrm p}\|\nu_{\mathrm p})
+
\mathcal I_{\mathrm g}(\mu,\nu).
}
\label{eq:final-chain}
\end{equation}

The additional term
\(\mathcal I_{\mathrm g}\) consists of three independent
types of auxiliary information: differences in conditional
gauge covariances, differences in conditional gauge means,
and differences in physical--gauge correlations.

The physical relative entropy is the minimum over Gaussian
auxiliary extensions:
\begin{equation}
D(\bar\mu\|\bar\nu)
=
\min_{\mu,\nu}
D(\mu\|\nu).
\label{eq:final-minimum}
\end{equation}

At the infinitesimal level, the Fisher information obeys
\begin{equation}
g^{\mathrm{aux}}
=
g^{\mathrm{phys}}
+
g^{\mathrm{cond}},
\qquad
g^{\mathrm{cond}}\geq0.
\label{eq:final-fisher}
\end{equation}

These identities do not assert that gauge variables possess
independent physical significance. On the contrary, they
quantify precisely the statistical information discarded
when passing from auxiliary electromagnetic potentials to
their gauge-equivalence classes.

The construction applies to finite Whitney complexes and
classical Gaussian fields. Continuum limits, quantum
relative entropy, non-Abelian gauge theories, and dynamical
evolution require separate analysis.

\vskip 12pt

\paragraph{\bf Data availability statement} No data is available for this work.

\vskip 12pt

\paragraph{\bf Conflict of interest statement} The author declares no conflict of interest.

\vskip 12pt

\paragraph{\bf Funding} No funding supported this work.

\vskip 12pt

\paragraph{\bf Acknowledgements} J.-P.M thanks the France 2030 framework programme Centre Henri Lebesgue ANR-11-LABX-0020-01 
for creating an attractive mathematical environment.

\vskip 12pt

\paragraph{\bf Author's Note on AI Assistance}
Portions of the text were developed with the assistance of a generative language model (OpenAI ChatGPT, based on the GPT-4 architecture). The AI was used to assist with drafting, editing, and standardizing the bibliography format. All mathematical content, structure, and theoretical constructions were provided, verified, and curated by the author. The author assumes full responsibility for the correctness, originality, and scholarly integrity of the final manuscript.

\end{document}